\documentclass[a4paper]{article}

\usepackage{geometry}

\usepackage{palatino}
\usepackage{latexsym}

\usepackage[mathscr]{eucal}
\usepackage{amsmath}
\usepackage{amsthm}
\usepackage{amsfonts}
\usepackage{amssymb}
\usepackage{amscd}

\usepackage{color}
\usepackage{graphicx}
\usepackage{graphics}
\usepackage{pifont}

\theoremstyle{plain}
\newtheorem{theorem}{Theorem}
\newtheorem{remark}{Remark}
\newtheorem{lemma}{Lemma}

\newcommand{\be}{\begin{equation}}
\newcommand{\ee}{\end{equation}}
\newcommand{\bey}{\begin{eqnarray}}
\newcommand{\eey}{\end{eqnarray}}

\input epsf

\newcommand{\donothing}[1]{}

\usepackage{bbm}

\newcommand{\R}{\mathbb{R}}

\title{Quantum confinement in $\alpha$-Grushin planes}

\begin{document}

\maketitle

\begin{center}
\author{Eugenio Pozzoli\footnote{Inria, Sorbonne Universit{\'e}, Universit{\'e} de Paris, CNRS, Laboratoire Jacques- Louis Lions, Paris, France (e-mail: eugenio.pozzoli@inria.fr)}}\;\;
\end{center}

\begin{abstract}
We consider here a family of singular Laplace-Beltrami operators, focussing our attention on the problem of so-called quantum confinement on the half-plane equipped with Riemannian metrics of Grushin type degenerate at the boundary. By introducing a costant-fiber direct integral scheme we are able to rigorously characterize the presence or absence of self-adjointness of these operators. We also compare our technique and results with the already studied problem of quantum confinement on the half-cylinder. 
\end{abstract}

\section*{Introduction}
The quantum evolution on a Riemannian oriented manifold $(M,g,\omega)$ can be studied via the Schrödinger equation

\begin{equation}\label{Schr}
i\dfrac{d}{dt}\psi=-\Delta_{g,\omega}\psi+V\psi, \qquad \psi(t)\in L^2(M,\omega).
\end{equation}

Several objects are involved in (\ref{Schr}): $g$ is a Riemannian metric on $M$, that is, a symmetric, positive definite $(0,2)$-tensor, $\omega$ is a volume form on $M$, that is, a non vanishing top differential form on $M$, $V$ is a real-valued function on $M$ (the potential) and $\Delta_{g,\omega}:= \text{div}_\omega \circ \text{grad}_g$ is a Laplace operator which acts on the Hilbert space $L^2(M,\omega)$ with domain $C_c^\infty(M)$, the compactly supported smooth functions on $M$, and which depends on both the metric and the volume.\\
Here we consider a geometric Laplacian, that is called the Laplace-Beltrami operator of $(M,g)$, which corresponds to the choice of the Riemannian volume form $\omega_g$. In this case, one can write $\Delta_{g,\omega_g}=:\Delta_g$. For a general introduction to the Laplace-Beltrami operator and geometric analysis see, e.g., \cite{Jost-RiemannGeom_GeomAnalysis}. \\
If the Schrödinger operator $H:=-\Delta_g+V$ is essentially self-adjoint on its domain $C_c^\infty(M)$, then the evolution is uniquely described by the propagator $e^{-it\overline{H}}$, $t \in \R$, where $\overline{H}$ is the closure of $H$, and also the unique self-adjoint extension of $H$.\\
If $H$ is not essentially self-adjoint on $C_c^\infty(M)$, it admits self-adjoint extensions since it is a real operator, i.e., it commutes with complex conjugation (\cite[Theorem X.3]{rs2}). In this case the extension is not unique and thus the quantum dynamics are not uniquely determined by $\Delta_g$ and $V$. In some sense, there is no natural quantum evolution. In what follows we will restrict ourselves to the case $V=0$.\\
For example, one can consider the usual euclidean Laplacian $\Delta=\sum_{j=1}^n \dfrac{\partial^2}{\partial x_j^2}$ on an open domain $A\subset \R^n$, without potential. Here, $\Delta$ is never essentially self-adjoint on $C_c^\infty(A)$, basically because it induces a Sobolev $H_0^2$ norm on the function space where it is defined. Its self-adjoint extensions correspond to the choice of suitable boundary conditions for the functions in the domain of $\Delta$.\\
The situation becomes more complex for domains $A$ of $\R^n$ equipped with non a non euclidean metric $g$. In this case, the geometric Laplacian $\Delta_g$ can be essentially self-adjoint on $C_c^\infty(M)$. This is basically due to some singularity of the Riemannian structure $(A,g)$ at the boundary $\partial A$ of $A$, which behaves like a confining \emph{effective potential}. When this occurs, we say that the particle is naturally \emph{confined} in $A$ during its quantum dynamics, since there exists a unique evolution $e^{-it\overline{\Delta_g}}$ on $L^2(A,\omega_g)$ without need of specifying boundary conditions. This is also called geometric confinement since it is only due to the geometric Laplacian $\Delta_g$ and not to a potential $V$. \\
So, the essential self-adjointness of $\Delta_g$ on $C_c^\infty(M)$ is usually called \emph{geometric quantum confinement} on $M$: no quantum information escapes from $M$. 
For further details, see for example \cite{Boscain-Prandi-JDE-2016,Boscain-Laurent-2013,Franceschi-Prandi-Rizzi-2017,GMP-Grushin-2018}.

\section{Grushin manifolds}
Here we are interested in studying the presence or absence of geometric quantum confinement in the family of Riemannian manifolds $M_\alpha:=(M,g_\alpha)$, $\alpha \in \R$,
$$M:=\{(x,y)\in \R^2| x>0\}=\R^+\times \R, \quad g_\alpha:=dx^2+\dfrac{1}{x^{2\alpha}}dy^2=\begin{pmatrix}
1 & 0 \\
0 &  1/x^{2\alpha}
\end{pmatrix}, $$
that is, the half-plane equipped with a family of Riemannian metrics $g_\alpha$.

The curvature of $M_\alpha$ is given by $K_\alpha(x)=-\dfrac{\alpha(\alpha+1)}{x^2}$, thus the Riemannian structure is degenerate at the boundary $\partial M$, that is, the $y$-axis, for every $\alpha$ except for $\alpha=0,1$. \\
The Riemannian volume form of $M_\alpha$ reads 
$$\omega_\alpha=\sqrt{\text{det}g_\alpha}dx\wedge dy=\dfrac{1}{x^\alpha}dx\wedge dy,$$
and the Laplace-Beltrami operator of $M_\alpha$ is
\begin{equation}\label{laplace}
\Delta_\alpha=\dfrac{\partial^2}{\partial x^2}+ x^{2\alpha}\dfrac{\partial^2}{\partial y^2}-\dfrac{\alpha}{x}\dfrac{\partial}{\partial x},
\end{equation}
acting on the Hilbert space $L^2\Big(\R^+\times \R,\dfrac{1}{x^\alpha}dxdy\Big)$, with domain $C_c^\infty(\R^+\times \R)$, the smooth functions on the half-plane compactly supported away from the $y$-axis.
In the work \cite{GMP-Grushin-2018} the essential self-adjointness of the operators (\ref{laplace}) is studied depending on $\alpha \in [0,\infty)$. These operators where first studied as examples of hypo-elliptic operators, and these manifolds find applications in sub-Riemannian geometry since they provide examples of quantum confinement in almost-Riemannian structure. For further details on almost-Riemannian and, more generally, sub-Riemannian structures, see, e.g., \cite{Montgomery-Subriemannian-2002,Agrachev-Barilari-Boscain-subriemannian}.  \\

Let us briefly recall how the self-adjointness of the operators (\ref{laplace}) can be analyse.\\
First of all, the operator $\Delta_\alpha$ is unitarily equivalent to
\begin{equation}\label{equivalent}
A_\alpha:= \dfrac{\partial^2}{\partial x^2}-\xi^2x^{2\alpha}-\dfrac{\alpha(2+\alpha)}{4x^2}
\end{equation}
 via two unitary tranformations:
$$U_\alpha:L^2\Big(\R^+\times \R,\dfrac{1}{x^\alpha}dxdy\Big) \Rightarrow  L^2(\R^+\times \R,dxdy),\quad \psi \mapsto \dfrac{1}{x^{\alpha/2}}\psi,$$
and 
$$\mathcal{F}_2:  L^2(\R^+\times \R,dxdy)\Rightarrow  L^2(\R^+\times \R,dxd\xi),$$
which is the Fourier tranform in the $y$ variable, that is,

$$\mathcal{F}_2(\psi)(x,\xi):=\dfrac{1}{\sqrt{2\pi}}\int_\R \psi(x,y)e^{-i\xi y}dy. $$

The domain of $A_\alpha$ is $\mathcal{D}(A_\alpha)=\mathcal{F}_2C_c^\infty(\R^+\times \R)$, that is, if $\psi \in \mathcal{D}(A_\alpha)$, then $\psi(\cdot,\xi)$ is compactly supported on $\R^+$, for every $\xi$, while $\psi(x,\cdot)$ is a Schwartz function on $\R$, for every $x$. This transformation is convenient since it allow us to decompose the action of the operator $A_\alpha$ in fiber actions. 

\section{Constant-fiber direct integral scheme}
More precisely, we identify the Hilbert spaces

$$L^2(\R^+\times \R,dxd\xi)\cong L^2(\R,d\xi;L^2(\R^+,dx))=: \int_\R^\oplus L^2(\R^+,dx)d\xi=:\mathcal{H}$$

thus we think of $L^2(\R^+\times \R,dxd\xi)$ as $L^2(\R^+,dx)$-valued square-integrable functions of $\xi \in \R$.

Consider, for any $\xi$, the fiber operator

$$A_\alpha(\xi):=\dfrac{d^2}{dx^2}-\xi^2x^{2\alpha}-\dfrac{\alpha(2+\alpha)}{4x^2}$$

which acts on the fiber space $L^2(\R^+,dx)$ with domain $C_c^\infty(\R^+)$. Then, let us define the operator

\begin{equation}\label{integrale}
B_\alpha:=\int_\R^\oplus \overline{A_\alpha(\xi)}d\xi, 
\end{equation}
whose action is defined by $(B_\alpha\psi)(\xi):=\overline{A_\alpha(\xi)}\psi(\xi)$, for any $\psi \in \mathcal{D}(B_\alpha)$, where the domain is
$$\mathcal{D}(B_\alpha):=\{\psi \in \mathcal{H}| \psi(\xi)\in \mathcal{D}(\overline{A_\alpha(\xi)}) \text{ for a.e. $\xi$, and } \int_\R \| \overline{A_\alpha(\xi)}\psi(\xi) \|_{L^2(\R^+,dx)}^2 d\xi <\infty \}. $$

Having introduced this new operator (\ref{integrale}), let us describe some of its properties and explain its relation with (\ref{equivalent}). 
\begin{theorem}\label{lavoro}
\begin{enumerate}.
\item $B_\alpha$ is a closed symmetric extension of $A_\alpha$.
\item If $A_\alpha(\xi)$ is essentially self-adjoint for a.e. $\xi \in \R$, then $B_\alpha$ is self-adjoint.
\item If $B_\alpha$ is self-adjoint, then $A_\alpha(\xi)$ is essentially self-adjoint for a.e. $\xi \in \R$.
\item If $B_\alpha$ is self-adjoint, then $A_\alpha^* \subset B_\alpha$.
\end{enumerate}
\end{theorem}
Property $2$ is well known, \cite[Theorem XIII.85(i)]{rs4}. The rest of Theorem \ref{lavoro} is proved in \cite{GMP-Grushin-2018}.\\
Finally one can deduce if $A_\alpha$ is essentially self-adjoint or not via the self-adjointness of the fiber operators $\overline{A_\alpha(\xi)}$.

\begin{theorem}\label{cruciale}
$A_\alpha$ is essentially self-adjoint if and only if $A_\alpha(\xi)$ is essentially self-adjoint for a.e. $\xi \in \R$.
\end{theorem}
\begin{proof}
We use the properties established in Theorem \ref{lavoro}. Assume that there exists a set $E \subset \R$ of strictly positive measure such that if $\xi \in E$ then $A_\alpha(\xi)$ is not essentially self-adjoint. Then, $B_\alpha$ is not self-adjoint. Moreover, one always has that $\overline{A_\alpha} \subset B_\alpha$. Now, if $A_\alpha$ was essentially self-adjoint, it could not be $\overline{A_\alpha} = B_\alpha$, because this would violate the lack of self-adjointness of $B_\alpha$. But it could not happen either that $B_\alpha$ is a proper extension of $\overline{A_\alpha}$, because self-adjoint operators are maximally symmetric. Therefore, $A_\alpha$ is not essentially self-adjoint.\\

Conversely, assume that $A_\alpha(\xi)$ is essentially self-adjoint for a.e. $\xi \in \R$. Then $B_\alpha$ is self-adjoint, and hence $A_\alpha^* \subset B_\alpha$. By taking the adjoint, we obtain $\overline{A_\alpha}=A_\alpha^{**}\supset B_\alpha^* =B_\alpha$, which, together with  $\overline{A_\alpha} \subset B_\alpha$, implies $\overline{A_\alpha} = B_\alpha$. The conclusion is that $A_\alpha$ is essentially self-adjoint. 
\end{proof}

\section{The compactification of $M_\alpha$}

In \cite{Boscain-Prandi-JDE-2016} quantum confinement has been studied for the cylindric version of $M_\alpha$, that is, 

$$\widetilde{M}:= \{(x,y)\in \R \times S^1| x>0\}=\R^+ \times S^1, \qquad \widetilde{M}_\alpha=(\widetilde{M},g_\alpha).$$

In this case, the Hilbert space is $ l^2(\mathbb{Z};L^2(\R^+,dx))$ and the inquired operator is

$$ \widetilde{A_\alpha}= \dfrac{\partial^2}{\partial x^2}-k^2x^{2\alpha}-\dfrac{\alpha(2+\alpha)}{4x^2}, \qquad k \in \mathbb{Z},$$

with domain $\mathcal{F}_2C_c^\infty(\R^+\times S^1)$, where 

$$\mathcal{F}_2(\psi)(x,k):=\dfrac{1}{\sqrt{2\pi}}\int_{S^1} \psi(x,\theta)e^{-ik \theta}d\theta.  $$

Let us introduce the fiber operators 

$$\widetilde{A_\alpha(k)}:= \dfrac{d^2}{dx^2}-k^2x^{2\alpha}-\dfrac{\alpha(2+\alpha)}{4x^2}, \qquad k \in \mathbb{Z},$$

as an operator on $L^2(\R^+,dx)$ with domain $C_c^\infty(\R^+)$. 

Hence the constant fiber direct integral reduces to a direct sum $\widetilde{B_\alpha}:= \bigoplus_{k \in \mathbb{Z}}\overline{\widetilde{A_\alpha(k)}}$

basically due to the discreteness of the spectrum of $\dfrac{\partial^2}{\partial \theta^2}$ as an operator on $L^2(S^1)$, with domain 

$$\mathcal{D}(B_\alpha):=\{(\psi_k)_{k \in \mathbb{Z}}\in l^2(\mathbb{Z};L^2(\R^+,dx) )| \psi_k \in \mathcal{D}(\overline{A_\alpha(k)}) \text{ for every k and } \sum_{k \in \mathbb{Z}} \|\overline{A_\alpha(k)}\psi_k \|_{L^2(\R^+,dx)}^2<\infty \}. $$

Actually in this compactified case one does not need all the previous construction of constant fiber direct integral. Indeed it suffices to define $\widetilde{C_\alpha}:= \bigoplus_{k\in \mathbb{Z}} \widetilde{A_\alpha(k)}$ as an operator on  $ l^2(\mathbb{Z};L^2(\R^+,dx))$ with domain all the sequences $L^2(\R^+,dx)$-valued which are zero except for a finite number of components. Then one has that $\overline{\widetilde{C_\alpha}}=\overline{\widetilde{A_\alpha}}$ and the essential self-adjointness of $\widetilde{C_\alpha}$ is controlled via the following lemma \cite[Proposition 2.3]{Boscain-Prandi-JDE-2016}.
\begin{lemma}\label{compatto}
The deficency indices of $\widetilde{C_\alpha}$ are $n_{\pm}(\widetilde{C_\alpha})=\sum_{k\in \mathbb{Z}} n_{\pm}(\widetilde{A_\alpha(k)})$.
\end{lemma}

\begin{remark}
Notice that this construction does not make sense for $M_\alpha$, since the domain of $B_\alpha$ made by sequences which are zero except for a finite numbers of components would contain just the function identically zero. Thus, we strongly need the constant fiber direct integral construction in the non-compact case.  
\end{remark}

Owning to the analysis on the fiber operators $\widetilde{A_\alpha(k)}$ via limit-point limit-circle method (see, e.g., \cite[Appendix to X.1]{rs2}), the fact that  $\overline{\widetilde{C_\alpha}}=\overline{\widetilde{A_\alpha}}$ and using Lemma \ref{compatto}, the description of quantum confinement on $\widetilde{M}_\alpha$ can be sum up in the following

\begin{theorem}\label{sum}
$\widetilde{A_\alpha}$ is essentially self-adjoint if and only if $\alpha \in (-\infty,-3]\cup [1,\infty)$. \\
When $\alpha \in (-3,1)$, the lack of self-adjointness is due to:
\begin{itemize}
\item all the Fourier $k$'s components if $\alpha \in (-1,1)$;
\item the only $0^{\text{th}}$ Fourier component if $\alpha \in (-3,-1]$.
\end{itemize}
\end{theorem}
In some sense, for $\alpha\in (-3,-1]$, the only information which escapes from $\R^+\times S^1$ is the average over $S^1$ of the wave function, that is, $\psi_0(x)=\dfrac{1}{2\pi}\int_{S^1}\psi(x,\theta)d\theta.$ For a proof of Theorem \ref{sum}, see \cite[Theorem 1.6]{Boscain-Prandi-JDE-2016}.

\section{Quantum confinement on $M_\alpha$}
We now develop explicitly the analysis of self-adjointness via limit-point limit-circle method on the fiber space $L^2(\R^+,dx)$ and then we integrate the result using Theorem \ref{cruciale} to obtain the characterization of quantum confinement on $M_\alpha$ depending on $\alpha$. We shall see how the constant fiber direct integral construction differs from a direct sum, obtaining indeed different regimes of $\alpha$, with respect to Theorem \ref{sum}.

\begin{theorem}
$A_\alpha$ is essentially self-adjoint if and only if $\alpha \in (-\infty,-1)\cup [1,\infty)$. \\
When $\alpha \in [-1,1)$, the lack of self-adjointness is due to:
\begin{itemize}
\item all the Fourier modes $\xi \in \R$ if $\alpha \in (-1,1)$;
\item the Fourier modes $\xi \in (-1,1)$ if $\alpha =-1$.
\end{itemize}
\end{theorem}
It is worth mentioning that, when $\alpha \in (-3,-1)$, the only Fourier mode which is not essentially self-adjoint corresponds to $\xi=0$, but this does not affect the self-adjointness of $B_\alpha$, because $\{0\}$ has zero measure and hence it does not affect the direct integral operator. We remark that the self-adjointness of $A_\alpha$, when $\alpha \geq 0$, was already classified in \cite[Corollary 3.8]{GMP-Grushin-2018}. Here we study also the case $\alpha<0$.
\begin{proof}
Thanks to Theorem \ref{cruciale}, $A_\alpha$ is essentially self-adjoint if and only if  $A_\alpha(\xi)$ is essentially self-adjoint as an operator on $L^2(\R^+,dx)$, for a.e. $\xi \in \R$. Now the self-adjointness of each $A_\alpha(\xi)$, $\xi \in \R$, can be studied by computing its deficency indices. Notice that  $-A_\alpha(\xi)$ is a Sturm-Liouville operator (see, e.g., \cite[Section 15.1]{schmu_unbdd_sa}), i.e., it has the form
$$ A_\alpha(\xi)= \dfrac{d^2}{dx^2}+V_{\alpha,\xi}, \qquad V_{\alpha,\xi}:= -\xi^2x^{2\alpha}-\dfrac{\alpha(2+\alpha)}{4x^2}.$$
Thus, by Weyl Theorem \cite[Theorem X.7]{rs2}, we can compute its deficency indices by checking if $V_{\alpha,\xi}$ is limit-point or limit-circle in $0$ and $\infty$. \\
Hence, we have to study the behaviour of the two linearly independent solutions of
\begin{equation}
\psi''(x)-\xi^2x^{2\alpha}\psi(x)-\dfrac{\alpha(2+\alpha)}{4x^2}\psi(x)=0
\end{equation}
near $0$ and $\infty$, that is, checking whether or not they are $L^2((0,1),dx)$ and $L^2((1,\infty),dx) $. \\

By applying a unitary transformation $U_\alpha:L^2(\R^+,dx) \Rightarrow L^2\Big(\R^+,\dfrac{1}{x^\alpha}dx\Big)$, $\psi \mapsto \dfrac{1}{x^{-\alpha/2}}\psi,$ this is equivalent to studying whether or not the two linearly independent solutions of

\begin{equation}\label{equazione}
\psi''(x)-\xi^2x^{2\alpha}\psi(x)-\dfrac{\alpha}{x}\psi'(x)=0
\end{equation}
are $L^2\Big((0,1),\dfrac{1}{x^\alpha}dx\Big)$ and $L^2\Big((1,\infty),\dfrac{1}{x^\alpha}dx\Big) $. \\
Let us start with $\xi=0$ and $\alpha \neq -1$, then the solutions to (\ref{equazione}) are $\psi_1=1$ and $\psi_2(x)=x^{1+\alpha}.$ They are both in $L^2\Big((0,1),\dfrac{1}{x^\alpha}dx\Big)$ if and only if $\alpha \in (-3,1)$. Moreover, $\psi_1,\psi_2 \notin L^2\Big((1,\infty),\dfrac{1}{x^\alpha}dx\Big) $ for every $\alpha\in \R$. Analogously, for $\xi=0$ and $\alpha=-1$ the solutions to (\ref{equazione}) are $\psi_1=1,\psi_2(x)=log(x)$, and both of them belong to $L^2\Big((0,1),x dx\Big)$, while no one of them belongs to $L^2\Big((1,\infty),x dx\Big) $. 
Thus, \begin{itemize}
\item if $\alpha \in (-3,1)$, then $V_{\alpha,0}$ is limit-circle in $0$ and limit-point in $\infty$;
\item if $\alpha \notin (-3,1)$, then $V_{\alpha,0}$ is limit-point in $0$ and $\infty$.
\end{itemize}
Consider now $\xi \neq 0$, and as before let us start with $\alpha \neq -1$. Then the solutions to (\ref{equazione}) are $\psi_1(x)=exp\Big(\dfrac{\xi x^{1+\alpha}}{1+\alpha}\Big),\psi_2(x)=exp\Big(-\dfrac{\xi x^{1+\alpha}}{1+\alpha}\Big) $. If $\alpha >-1$, $\psi_1$ and $\psi_2$ are bounded near $0$, and hence $0$ is limit-circle when $\alpha < 1$. One of them grows exponentially at $\infty$, and hence $\infty$ is always limit-point for $\alpha >-1$.\\
If $\alpha<-1$, $\psi_1$ and $\psi_2$ converge to $1$ as $x\Rightarrow \infty$ and the measure $\dfrac{1}{x^\alpha}dx$ explodes, and hence $\infty$ is always limit-point when $\alpha<-1$, while one of them diverges exponentially in $0$, and hence $0$ is always limit-point when $\alpha<-1$. \\
If $\alpha=-1$, the solutions of (\ref{equazione}) are $\psi_1(x)=x^\xi$, $\psi_2(x)=x^{-\xi}$. One of them diverges when $x \Rightarrow \infty$ and the measure $x dx$ diverges, hence $\infty$ is limit-point, while $\psi_1,\psi_2 \in L^2((0,1),xdx)$ if and only if $\xi \in (-1,1)$, that is, $0$ is limit-circle if $\xi \in (-1,1)$ and limit-point if $\xi \notin (-1,1)$.\\
Summing up, and using that $A_\alpha(\xi)$ is essentially self-adjoint if and only if $0$ and $\infty$ are limit-point for $V_{\alpha,\xi}$, that is, the Weyl Theorem, we obtain
\begin{itemize}
\item if $\alpha \geq 1$, then $A_\alpha(\xi)$ is essentially self-adjoint, for every $\xi \in \R$;
\item if $-1<\alpha<1$, then $A_\alpha(\xi)$ is not essentially self-adjoint, for every $\xi \in \R$;
\item if $\alpha=-1$, then $A_\alpha(\xi)$ is not essentially self-adjoint for $\xi \in (-1,1)$ and it is essentially self-adjoint if $\xi \notin (-1,1)$;
\item if $-3<\alpha<-1$, then $A_\alpha(0)$ is not essentially self-adjoint and $A_\alpha(\xi)$ is essentially self-adjoint for every $\xi \neq 0$;
\item if $\alpha \leq -3$, then $A_\alpha(\xi)$ is essentially self-adjoint for every $\xi \in \R$.  
\end{itemize}\end{proof}

\textbf{Acknowledgement} The project leading to this publication has received funding from the European Union’s Horizon 2020 research and innovation programme under the Marie Sklodowska-Curie grant agreement no. 765267 (QuSCo). Also it was supported by the ANR projects SRGI ANR-15-CE40-0018 and Quaco ANR-17-CE40-0007-01.

\bibliographystyle{amsalpha}
\bibliography{references}

\providecommand{\bysame}{\leavevmode\hbox to3em{\hrulefill}\thinspace}
\providecommand{\MR}{\relax\ifhmode\unskip\space\fi MR }
\providecommand{\MRhref}[2]{%
  \href{http://www.ams.org/mathscinet-getitem?mr=#1}{#2}
}
\providecommand{\href}[2]{#2}
\begin{thebibliography}{GMP19}

\bibitem[ABB20]{Agrachev-Barilari-Boscain-subriemannian}
Andrei Agrachev, Davide Barilari, and Ugo Boscain, \emph{A comprehensive
  introduction to sub-{R}iemannian geometry}, Cambridge Studies in Advanced
  Mathematics, vol. 181, Cambridge University Press, Cambridge, 2020, From the
  Hamiltonian viewpoint, With an appendix by Igor Zelenko. \MR{3971262}

\bibitem[BL13]{Boscain-Laurent-2013}
Ugo Boscain and Camille Laurent, \emph{{The {L}aplace-{B}eltrami operator in
  almost-{R}iemannian geometry}}, Ann. Inst. Fourier (Grenoble) \textbf{63}
  (2013), no.~5, 1739--1770.

\bibitem[BP16]{Boscain-Prandi-JDE-2016}
Ugo Boscain and Dario Prandi, \emph{{Self-adjoint extensions and stochastic
  completeness of the {L}aplace-{B}eltrami operator on conic and anticonic
  surfaces}}, J. Differential Equations \textbf{260} (2016), no.~4, 3234--3269.

\bibitem[FPR20]{Franceschi-Prandi-Rizzi-2017}
Valentina Franceschi, Dario Prandi, and Luca Rizzi, \emph{On the essential
  self-adjointness of singular sub-{L}aplacians}, Potential Anal. \textbf{53}
  (2020), no.~1, 89--112. \MR{4117982}

\bibitem[GMP19]{GMP-Grushin-2018}
Matteo Gallone, Alessandro Michelangeli, and Eugenio Pozzoli, \emph{On
  geometric quantum confinement in {G}rushin-type manifolds}, Z. Angew. Math.
  Phys. \textbf{70} (2019), no.~6, Paper No. 158, 17. \MR{4019735}

\bibitem[Jos17]{Jost-RiemannGeom_GeomAnalysis}
J{\"u}rgen Jost, \emph{{Riemannian geometry and geometric analysis}}, seventh
  ed., {Universitext}, Springer, Cham, 2017.

\bibitem[Mon02]{Montgomery-Subriemannian-2002}
Richard Montgomery, \emph{{A tour of subriemannian geometries, their geodesics
  and applications}}, {Mathematical Surveys and Monographs}, vol.~91, American
  Mathematical Society, Providence, RI, 2002.

\bibitem[RS75]{rs2}
Michael Reed and Barry Simon, \emph{{Methods of modern mathematical physics.
  {II}. {F}ourier analysis, self-adjointness}}, Academic Press [Harcourt Brace
  Jovanovich, Publishers], New York-London, 1975. \MR{0493420 (58 \#12429b)}

\bibitem[RS78]{rs4}
\bysame, \emph{{Methods of modern mathematical physics. {IV}. {A}nalysis of
  operators}}, Academic Press [Harcourt Brace Jovanovich, Publishers], New
  York-London, 1978.

\bibitem[Sch12]{schmu_unbdd_sa}
Konrad Schm{\"u}dgen, \emph{{Unbounded self-adjoint operators on {H}ilbert
  space}}, {Graduate Texts in Mathematics}, vol. 265, Springer, Dordrecht,
  2012. \MR{2953553}

\end{thebibliography}

\end{document}